\documentclass[conference]{IEEEtran}
\IEEEoverridecommandlockouts
\usepackage{cite}
\usepackage{amsmath,amssymb,amsfonts}
\usepackage{graphicx}
\usepackage{textcomp}
\usepackage{xcolor}
\def\BibTeX{{\rm B\kern-.05em{\sc i\kern-.025em b}\kern-.08em
    T\kern-.1667em\lower.7ex\hbox{E}\kern-.125emX}}

\usepackage{amsthm}
\usepackage[ruled,vlined]{algorithm2e}
\usepackage{booktabs}
\usepackage{subfigure}
\newcommand{\argmax}{\mathop{\rm argmax}}

\theoremstyle{plain}
\newtheorem{theorem}     {Theorem}
\newtheorem{lemma}       {Lemma}

\newtheorem{problem}     {Problem}

\newtheorem{fact}        {Fact}

\begin{document}

\title{Robust Densest Subgraph Discovery
}


\author{\IEEEauthorblockN{Atsushi Miyauchi}
\IEEEauthorblockA{RIKEN AIP}
Tokyo, Japan\\
atsushi.miyauchi.hv@riken.jp
\and
\IEEEauthorblockN{Akiko Takeda}
\IEEEauthorblockA{The University of Tokyo / RIKEN AIP} 
Tokyo, Japan\\
takeda@mist.i.u-tokyo.ac.jp
}

\maketitle

\begin{abstract}
Dense subgraph discovery is an important primitive in graph mining, which has a wide variety of applications in diverse domains. 
In the densest subgraph problem, given an undirected graph $G=(V,E)$ with an edge-weight vector $w=(w_e)_{e\in E}$, 
we aim to find $S\subseteq V$ that maximizes the density, i.e., $w(S)/|S|$, where $w(S)$ is the sum of the weights of the edges in the subgraph induced by $S$. 
Although the densest subgraph problem is one of the most well-studied optimization problems for dense subgraph discovery, 
there is an implicit strong assumption; it is assumed that the weights of all the edges are known exactly as input. 
In real-world applications, there are often cases where we have only \emph{uncertain} information of the edge weights. 
In this study, we provide a framework for dense subgraph discovery under the uncertainty of edge weights. 
Specifically, we address such an uncertainty issue using the theory of \emph{robust optimization}. 
First, we formulate our fundamental problem, the \emph{robust densest subgraph problem}, and present a simple algorithm. 
We then formulate the \emph{robust densest subgraph problem with sampling oracle} 
that models dense subgraph discovery using an edge-weight sampling oracle, 
and present an algorithm with a strong theoretical performance guarantee. 
Computational experiments using both synthetic graphs and popular real-world graphs demonstrate 
the effectiveness of our proposed algorithms. 
\end{abstract}


\section{Introduction}~\label{sec:introduction}
Dense subgraph discovery, or extracting a dense component in a graph, is an important primitive in graph mining, 
which has a wide variety of applications in diverse domains. 
A typical application is the identification of components 
that have certain special roles or possess important functions in underlying systems represented by graphs. 
For example, consider the protein--protein interaction graphs, 
where vertices represent the proteins within a cell and 
edges (resp.~edge weights) represent the interactions (resp.~strength of interactions) among the proteins. 
The dense components in this graph are likely to be the sets of proteins 
that exhibit identical or similar functions within the cell~\cite{Bader_Hogue_03}. 
As another example, consider the Web graph, where vertices represent web pages and edges represent the hyperlinks among them. 
The dense components in this graph are generally communities 
(i.e., the set of web pages addressing identical or similar topics)~\cite{Dourisboure+_07} 
and occasionally spam link farms~\cite{Gibson+_05}, which are effective for improving Web search engines. 
Other application examples include identifying regulatory motifs in DNA~\cite{Fratkin+_06}, 
decision-making for cost-effective marketing strategies~\cite{Miyauchi+_15}, 
expert team formation~\cite{Bonchi+_14,Tsourakakis+_13}, 
and real-time story identification in micro-blogging streams~\cite{Angel+_12}. 

The densest subgraph problem is one of the most well-studied optimization problems for dense subgraph discovery. 
Let $G=(V,E)$ be an undirected graph with an edge-weight vector $w=(w_e)_{e\in E}$. 
We denote by $G[S]$ the subgraph induced by $S\subseteq V$, i.e., $G[S]=(S,E(S))$, 
where $E(S)=\{\{u,v\}\in E\mid u,v\in S\}$.  
For an edge-weight vector $w=(w_e)_{e\in E}$, 
the \emph{density} of $S\subseteq V$ is defined as $f_w(S)=w(S)/|S|$, 
where $w(S)$ is the sum of the weights of the edges in $G[S]$, i.e., $w(S)=\sum_{e\in E(S)}w_e$. 
In the \emph{(weighted) densest subgraph problem}, given an undirected graph $G=(V,E)$ with an edge-weight vector $w=(w_e)_{e\in E}$, 
we aim to find $S\subseteq V$ that maximizes the density $f_w(S)=w(S)/|S|$. 
An optimal solution is called a \emph{densest subgraph}. 

The densest subgraph problem has recently attracted significant interest 
because it can be solved exactly in polynomial time and with adequate approximation in almost linear time. 
There are exact algorithms such as Goldberg's flow-based algorithm~\cite{Goldberg_84} and Charikar's LP-based algorithm~\cite{Charikar_00}. 
Moreover, Charikar~\cite{Charikar_00} demonstrated that the greedy peeling algorithm designed by Asahiro et al.~\cite{Asahiro+_00} 
is a $1/2$-approximation algorithm\footnote{
A feasible solution is said to be \emph{$\alpha$-approximate} 
if its objective value is greater than or equal to the optimal value times $\alpha$. 
An algorithm is called an \emph{$\alpha$-approximation algorithm} 
if it runs in polynomial time and returns an $\alpha$-approximate solution for any instances. 
} 
for the problem. 
This can be implemented to run in $O(m+n\log n)$ time for weighted graphs and $O(m+n)$ time for unweighted graphs, 
where $n=|V|$ and $m=|E|$. 


However, in the densest subgraph problem, there is an implicit strong assumption; 
it is assumed that the weights of all the edges are known exactly as input. 
In numerous real-world applications, there are often cases where we have only \emph{uncertain} information of the edge weights. 
For example, consider the protein--protein interaction graphs. 
In the generation process of such graphs, the edge weights representing the strength of the interactions among the proteins 
are commonly obtained through biological experiments using measuring instruments with some noises. 
In such a scenario, we have only the estimated values for \emph{true} edge weights. 
Therefore, it is challenging to provide a framework for dense subgraph discovery under the uncertainty of edge weights.

\subsection{Our Contribution}
In this study, we provide a framework for dense subgraph discovery under the uncertainty of edge weights. 
Specifically, we address such an uncertainty issue using the theory of \emph{robust optimization}.


To model the uncertainty of edge weights in real-world applications, 
we assume that we have only an \emph{edge-weight space} $W=\times_{e\in E}[l_e,r_e]\subseteq \times_{e\in E}[0,\infty)$ 
(rather than an edge-weight vector $w=(w_e)_{e\in E}$) 
that contains the \emph{unknown} true edge-weight vector $w^\text{true}=(w^\text{true}_e)_{e\in E}$. 
The edge-weight space can be considered as a product of the confidence intervals of the true edge weights, 
each of which (i.e., $[l_e,r_e]$ for $e\in E$) can be obtained in practice from theoretically guaranteed lower and upper bounds on the true edge weight 
or repeated sampling of an estimated value of the true edge weight. 

The key question is as follows: In this uncertain situation, how can we evaluate the quality of $S\subseteq V$?
Note here that as we know nothing about $w^\text{true}$ apart from the fact that $w^\text{true}\in W$,
we cannot directly use the value of $f_{w^\text{true}}(S)$ for evaluating $S$. 
To answer the question, we use a well-known concept in the theory of robust optimization. 
In the robust optimization paradigm, the quality of a solution for a robust optimization problem 
is generally evaluated using a measure called the \emph{robust ratio}. 
In our scenario, the robust ratio of $S\subseteq V$ under edge-weight space $W$ is defined as 
the multiplicative gap between the density of $S$ (i.e., $f_{w'}(S)$) and the density of $S^*_{w'}$ (i.e., $f_{w'}(S^*_{w'})$)
under the worst-case edge-weight vector $w'\in W$, where $S^*_{w'}$ is an optimal solution to the densest subgraph problem on $G$ with $w'$. 
Intuitively, $S\subseteq V$ with a large robust ratio has a density close to the optimal value 
even on $G$ with the edge-weight vector selected adversarially from $W$. 
Using the robust ratio, we formulate the \emph{robust densest subgraph problem} as follows:
Given an undirected graph $G=(V,E)$ with an edge-weight space $W=\times_{e\in E}[l_e,r_e]$, 
we aim to find $S\subseteq V$ that maximizes the robust ratio under $W$. 

For the robust densest subgraph problem, we first provide a strong negative result; 
specifically, we show that there exist some instances $G=(V,E)$ with $W=\times_{e\in E}[l_e,r_e]$ 
for which any (deterministic) algorithm returns $S\subseteq V$ that has a robust ratio of $O(1/n)$. 
Then, in contrast to this negative result, we present a simple algorithm that utilizes an exact algorithm for the (original) densest subgraph problem. 
We demonstrate that for any instance that satisfies $\min_{e\in E}l_e>0$, 
our algorithm returns $S\subseteq V$ that has a robust ratio of at least $\frac{1}{\min_{e\in E}\frac{r_e}{l_e}}$. 
Moreover, we prove that the lower bound on the robust ratio achieved by our proposed algorithm 
is the best possible except for the constant factor. 

The lower bound on the robust ratio achieved by our algorithm (i.e., $\frac{1}{\min_{e\in E}\frac{r_e}{l_e}}$) is still small, 
although it is the best possible except for the constant factor. 
This negative result was caused by the fact that 
in the robust densest subgraph problem, we were \emph{excessively conservative} in evaluating the quality of $S\subseteq V$, 
that is, we aimed to find $S\subseteq V$ that has a relatively large density compared to the optimal value on $G$ 
with \emph{any} edge-weight vector $w\in W$. 
In some real-world applications, 
each confidence interval (i.e., $[l_e,r_e]$ for $e\in E$) may be obtained from repeated sampling of an estimated value of the true edge weight; 
therefore, we conjecture that we can obtain a significantly better lower bound on the robust ratio by using such samplings more sophisticatedly. 

To this end, we formulate the \emph{robust densest subgraph problem with sampling oracle} as follows: 
We are given an undirected graph $G=(V,E)$ with an edge-weight space $W=\times_{e\in E}[l_e,r_e]$, 
wherein the unknown true edge-weight vector $w^\text{true}=(w^\text{true}_e)_{e\in E}$ exists. 
In addition, we have access to an edge-weight sampling oracle
that accepts an edge $e\in E$ as input and returns a real value as output, in time $\theta$, 
that was drawn independently from a distribution on $[l_e,r_e]$ 
in which the expected value is equal to the true edge weight $w^\text{true}_e$. 
Given $\gamma \in (0,1)$, we aim to find $W_\text{out}\subseteq W$ 
that satisfies $w^\text{true}\in W_\text{out}$ with a probability of at least $1-\gamma$
and $S_\text{out}\subseteq V$ that maximizes the robust ratio under $W_\text{out}$. 
An important fact is that if we obtain $S_\text{out}$ with an objective function value of $\alpha$, 
the subset $S_\text{out}$ is an $\alpha$-approximate solution for the densest subgraph problem on $G$ with $w^\text{true}$, 
with a probability of at least $1-\gamma$. 

For the robust densest subgraph problem with sampling oracle, we present an algorithm with a strong theoretical performance guarantee. 
Specifically, for any $\gamma \in (0,1)$ and $\epsilon >0$, our algorithm obtains 
$W_\text{out}\subseteq W$ that satisfies $w^\text{true}\in W_\text{out}$ with a probability of at least $1-\gamma$ 
and $S_\text{out}\subseteq V$ that has a robust ratio of at least $1-\epsilon$ under the edge-weight space $W_\text{out}$, 
in time pseudo-polynomial in the size of $G$ and $W$, $\theta$, and $1/\epsilon$. 
Therefore, we observe that our algorithm obtains a $(1-\epsilon)$-approximate solution for the densest subgraph problem on $G$ with $w^\text{true}$, 
with a probability of at least $1-\gamma$. 

Finally, we conduct computational experiments to evaluate the effectiveness of our proposed algorithms 
in terms of both the quality of solutions and computation time. 
We compare our proposed algorithms with a certain baseline algorithm 
using both synthetic graphs and popular real-world graphs. 
To generate synthetic graphs appropriate for our experimental evaluation, we introduce a random graph model, 
which we refer to as the \emph{planted uncertain dense subgraph model}. 
With regard to real-world graphs, we introduce a random model for constructing an edge-weight space and a true edge-weight vector for a given graph, 
which we refer to as the \emph{knockout densest subgraph model}. 
The results demonstrate the effectiveness of our proposed algorithms.


\subsection{Related Work}\label{subsec:related}
Robust optimization, which has been actively studied in the field of operations research, 
is known to be an effective methodology for addressing optimization problems  
under uncertainty~\cite{Ben-Tal+09,Ben-Tal_Nemirovski_98,Ben-Tal_Nemirovski_99}. 
Recently, the theory of robust optimization has been widely applied to 
tasks in knowledge discovery and data mining, particularly to graph mining tasks. 
For example, Chen et al.~\cite{Chen+16} and He and Kempe~\cite{He_Kempe_16} studied \emph{robust influence maximization}, 
which is a robust variation of the popular graph mining task called \emph{influence maximization}. 
Their focus was on the influence maximization counterpart of our work; 
they aimed to find a subset of vertices that exhibits a large robust ratio in terms of the influence. 
In particular, Chen et al.~\cite{Chen+16} developed an algorithm with a theoretical performance guarantee using a certain sampling oracle. 
To the best of our knowledge, we are the first to utilize the theory of robust optimization 
for addressing dense subgraph discovery under uncertainty. 

Apart from the uncertainty of edge weights, 
a large body of work has been devoted to graph mining tasks with the uncertainty of the \emph{existence} of edges. 
In this scenario, it is generally assumed that we are given an \emph{uncertain graph}, i.e., 
a graph $G=(V,E)$ with a function $p:E\rightarrow [0,1]$ 
in which $e\in E$ is present with probability $p(e)$ whereas $e\in E$ is absent with probability $1-p(e)$. 
For a number of fundamental optimization problems on graphs, their counterparts on uncertain graphs have been introduced~\cite{Kassiano+16}. 
In particular, Zou~\cite{Zou_13} studied the densest subgraph problem on uncertain graphs. 
In this problem, given an uncertain graph $G=(V,E)$ with a function $p:E\rightarrow [0,1]$, 
we are asked to find $S\subseteq V$ that maximizes the expected value of the density. 
Zou~\cite{Zou_13} demonstrated that this problem can be reduced to the (original weighted) densest subgraph problem 
and developed a polynomial-time exact algorithm using the reduction. 
It should be noted that the problems we formulate in the present study cannot be addressed using uncertain graphs. 
In fact, uncertain graphs do not consider the uncertainty of edge weights; they only model the uncertainty of the existence of edges. 

In addition to the variant on uncertain graphs, the densest subgraph problem has numerous noteworthy problem variations. 
Examples include the size-constraint variants~\cite{Andersen_Chellapilla_09,Bhaskara+_10,Feige+01,Khuller_Saha_09,Papailiopoulos+14} 
and the variants generalizing the term $w(S)$ in the density~\cite{Mitzenmacher+15,Miyauchi_Kakimura_18,Tsourakakis_15} 
and the term $|S|$ in the density~\cite{Kawase_Miyauchi_17}. 
Furthermore, a large body of work has been devoted to the streaming or dynamic settings of the densest subgraph problem~\cite{Bahmani+_12,Bhattacharya+_15,Epasto+_15,Hu+_17,McGregor+_15,Nasir+_17}. 
Some literatures have considered the densest subgraph problem 
on hypergraphs~\cite{Hu+_17,Miyauchi+_15} or on multilayer networks~\cite{Galimberti+_17}.

\subsection{Paper Organization}\label{subsec:organization}
In Section~\ref{sec:preliminaries}, we revisit some existing algorithms for the densest subgraph problem, 
which will be used in the design of our proposed algorithms. 
In Section~\ref{sec:robust}, we formulate the robust densest subgraph problem and present a simple algorithm. 
Then, in Section~\ref{sec:robust_sampling}, we formulate the robust densest subgraph problem with sampling oracle 
and present an algorithm with a strong theoretical performance guarantee. 
We report the results of our computational experiments in Section~\ref{sec:experiments}. 
We conclude the study in Section~\ref{sec:conclusion}. 

\section{Preliminaries}\label{sec:preliminaries}

Here, we describe Charikar's LP-based exact algorithm for the densest subgraph problem~\cite{Charikar_00}, 
which will be used in the design of our proposed algorithms. 
The algorithm introduces a variable $x_e$ for each $e\in E$ and a variable $y_v$ for each $v\in V$, 
and solves the following LP in polynomial time:  
\begin{alignat*}{3}
&\text{maximize}   &\ \ \, &\sum_{e\in E}w_ex_e\\
&\text{subject to} &       &x_e\leq y_u,\ x_e\leq y_v  &\quad &\forall e=\{u,v\}\in E,\\
&                  &       &\sum_{v\in V}y_v=1,\\
&                  &       &x_e,\,y_v\geq 0            &      &\forall e\in E,\,\forall v\in V.
\end{alignat*}
Intuitively, this LP is a standardized (i.e., linearized) version of a continuous relaxation of the original problem. 
Let $(x^*,y^*)$ be an optimal solution to this LP. 
For a real parameter $r\geq 0$, the algorithm introduces a sequence of subsets of vertices 
$S(r)=\{v\in V\mid y^*_v\geq r\}$ 
and finds $r^*\in \argmax_{r\in [0,1]} f_w(S(r))$. 
It should be noted that such $r^*$ can be found by simply examining $r=y^*_v$ for each $v\in V$. 
Finally, the algorithm returns $S(r^*)$. 
Charikar~\cite{Charikar_00} established that the output of the algorithm, i.e., $S(r^*)$, 
is an optimal solution to the densest subgraph problem. 

The above LP-based algorithm is elegant and convenient to implement (if we use a mathematical programming solver such as Gurobi Optimizer or IBM ILOG CPLEX); 
however, in practice, it is applicable only to graphs with a maximum of hundreds of thousands of edges. 
Recently, Balalau et al.~\cite{Balalau+_15} developed a highly effective preprocessing algorithm for the densest subgraph problem. 
Their preprocessing algorithm first runs the greedy peeling algorithm 
to obtain a $1/2$-approximate solution $S_\text{approx}\subseteq V$. 
Specifically, the greedy peeling algorithm iteratively removes a vertex with the smallest (weighted) degree 
in a current remaining graph to obtain a sequence of subsets from $V$ to $\emptyset$ 
and returns the best subset among the sequence. 
Then, the preprocessing algorithm removes every vertex whose weighted degree is strictly less than $f_w(S_\text{approx})$. 
Balalau et al.~\cite{Balalau+_15} indicated that this preprocessing does not remove any vertex contained in $S^*\subseteq V$, 
where $S^*$ is an arbitrary optimal solution to the densest subgraph problem. 
Therefore, whenever we wish to obtain an optimal solution to the densest subgraph problem, 
we can apply Balalau et al.'s preprocessing to the input. 
It should be noted that in practice, 
Charikar's LP-based algorithm in combination with Balalau et al.'s preprocessing  
can obtain an optimal solution in reasonable time (i.e., a few tens of minutes) even on graphs with a few millions of edges. 

\section{Robust Densest Subgraph Problem}\label{sec:robust}
In this section, we formulate the robust densest subgraph problem and present a simple algorithm.

\subsection{Problem Definition}\label{subsec:robust_problem}

To model the uncertainty of edge weights in real-world applications, 
we assume that we have only an \emph{edge-weight space} $W=\times_{e\in E}[l_e,r_e]\subseteq \times_{e\in E}[0,\infty)$ 
(rather than an edge-weight vector $w=(w_e)_{e\in E}$) 
that contains the \emph{unknown} true edge-weight vector $w^\text{true}=(w^\text{true}_e)_{e\in E}$. 
As we know nothing about $w^\text{true}$ except for the fact $w^\text{true}\in W$,
we cannot directly use the value of $f_{w^\text{true}}(S)$ for evaluating $S\subseteq V$. 
Here, we use a well-known concept in the theory of robust optimization, which is called the \emph{robust ratio}. 
In our scenario, the robust ratio of $S\subseteq V$ under edge-weight space $W$ is defined as
\begin{align*}
\min_{w\in W}\frac{f_w(S)}{f_w(S^*_w)}, 
\end{align*}
where $S^*_w\subseteq V$ is a densest subgraph on $G$ with edge-weight vector $w$. 
Intuitively, $S\subseteq V$ with a large robust ratio has a density close to the optimal value 
even on $G$ with the edge-weight vector selected adversarially from $W$. 
Using the robust ratio, we formulate the \emph{robust densest subgraph problem} as follows: 
\begin{problem}[Robust densest subgraph problem]\label{prob:rds}
Given an undirected graph $G=(V,E)$ with an edge-weight space $W=\times_{e\in E} [l_e,r_e]\subseteq \times_{e\in E}[0,\infty)$, 
we are asked to find a subset of vertices $S\subseteq V$ that maximizes the robust ratio under edge-weight space $W$: 
\begin{align*}
\min_{w\in W}\frac{f_w(S)}{f_w(S^*_w)}, 
\end{align*}
where $S^*_w\subseteq V$ is a densest subgraph on $G$ with edge-weight vector $w$. 
\end{problem}

This problem is a generalization of the (original) densest subgraph problem. 
In fact, if $l_e=r_e$ holds for every $e\in E$, the problem reduces to the densest subgraph problem. 

Unfortunately, we have the following strong negative result for the robust densest subgraph problem. 
\begin{theorem}\label{thm:negative}
There exists an instance of the robust densest subgraph problem (Problem~\ref{prob:rds})
for which any (deterministic) algorithm returns $S\subseteq V$ that has a robust ratio of $O(1/n)$. 
\end{theorem}
\begin{proof}
Let $G=(V,E)$ be any graph in which every vertex has degree of at least one. 
We take $W=\times_{e\in E}[l_e,r_e]$ that satisfies $l_e=0$ and $r_e>0$ for each $e\in E$. 
Note that any deterministic algorithm for Problem~\ref{prob:rds} returns some $S\subseteq V$. 

For any $S\subsetneq V$, there exists an edge $e'\in E\setminus E(S)$. 
We can construct an edge weight $w'=(w'_e)_{e\in E}$ such that for each $e\in E$, 
\begin{align*}
w'_e=
\begin{cases}
r_e  &\text{if } e=e',\\
0  &\text{otherwise}. 
\end{cases}
\end{align*}
As $w'\in W$ holds, the robust ratio of $S$ can be upper bounded as follows: 
\begin{align*}
\min_{w\in W}\frac{f_w(S)}{f_w(S^*_w)}
\leq \frac{f_{w'}(S)}{f_{w'}(S^*_{w'})}
= 0. 
\end{align*}

On the other hand, let $S=V$. 
Let us select an arbitrary edge $e'\in E$. 
We can again construct an edge weight $w'=(w'_e)_{e\in E}$ such that for each $e\in E$, 
\begin{align*}
w'_e=
\begin{cases}
r_e  &\text{if } e=e',\\
0  &\text{otherwise}. 
\end{cases}
\end{align*}
As $w'\in W$ again holds, the robust ratio of $S\ (=V)$ can be upper bounded as follows: 
\begin{align*}
\min_{w\in W}\frac{f_w(S)}{f_w(S^*_w)}
\leq \frac{f_{w'}(S)}{f_{w'}(S^*_{w'})} 
\leq \frac{r_e/n}{r_e/2}
=O\left(\frac{1}{n}\right). 
\end{align*}
Thus, we have the theorem. 
\end{proof}

\subsection{Algorithm and Analysis}\label{subsec:robust_algorithm}
In contrast to the above negative result, we now present a simple algorithm for the robust densest subgraph problem, 
which utilizes an exact algorithm for the (original) densest subgraph problem. 
Let $w^-=(l_e)_{e\in E}$ and $w^+=(r_e)_{e\in E}$. 
Our algorithm computes $S^*_{w^-}\subseteq V$, i.e., a densest subgraph on $G$ with extreme edge weight $w^-=(l_e)_{e\in E}$ and returns it. 
For reference, the procedure is described in Algorithm~\ref{alg:basic}. 
\begin{algorithm}[t]
\caption{Basic algorithm}\label{alg:basic}
\SetKwInOut{Input}{Input} 
\SetKwInOut{Output}{Output} 
\Input{\ $G=(V,E)$ with $W=\times_{e\in E}[l_e,r_e]$}
\Output{\ $S_\text{out}\subseteq V$}
$S_\text{out}\leftarrow~\text{A densest subgraph on } G \text{ with edge-weight vector }$\\
\hspace{10.5mm}$w^-=(l_e)_{e\in E}$\; 
\Return $S_\text{out}$\;
\end{algorithm}

In the following, we provide the theoretical performance guarantee of Algorithm~\ref{alg:basic}. 
To this end, we use the following lemma, 
which provides the fundamental property of the density function, 
i.e., the monotonicity of $f_w(S)$ with respect to edge-weight vector $w$. 
The proof is straightforward and therefore omitted. 
\begin{lemma}\label{lem:monotonicity}
Let $G=(V,E)$ be an undirected graph. 
Let $w_1$ and $w_2$ be edge-weight vectors such that $w_1\leq w_2$ holds. 
Then, for any $S\subseteq V$, it holds that $f_{w_1}(S)\leq f_{w_2}(S)$. 
\end{lemma}

The following theorem provides the theoretical performance guarantee of Algorithm~\ref{alg:basic}. 
More specifically, the theorem presents a lower bound on the robust ratio of the output of Algorithm~\ref{alg:basic} 
under a certain reasonable condition. 
\begin{theorem}\label{thm:robust_positive}
Let $G=(V,E)$ with $W=\times_{e\in E}[l_e,r_e]$ be an instance of the robust densest subgraph problem (Problem~\ref{prob:rds}). 
Suppose that $\min_{e\in E}l_e>0$ holds. 
Then, Algorithm~\ref{alg:basic} returns $S\subseteq V$ 
that has a robust ratio of at least $\frac{1}{\max_{e\in E}\frac{r_e}{l_e}}$. 
\end{theorem}
\begin{proof}
Recall that Algorithm~\ref{alg:basic} returns $S^*_{w^-}$, 
i.e., a densest subgraph on $G$ with extreme edge weight $w^-=(l_e)_{e\in E}$. 
Then, we can lower bound the robust ratio as follows: 
\begin{align*}
\min_{w\in W}\frac{f_w(S^*_{w^-})}{f_w(S^*_{w})}
&\geq \frac{f_{w^-}(S^*_{w^-})}{f_{w^+}(S^*_{w^+})}
\geq \frac{f_{w^-}(S^*_{w^+})}{f_{w^+}(S^*_{w^+})}\\
&\geq \frac{1}{\max_{e\in E}\frac{r_e}{l_e}}\cdot \frac{f_{w^+}(S^*_{w^+})}{f_{w^+}(S^*_{w^+})}
= \frac{1}{\max_{e\in E}\frac{r_e}{l_e}},
\end{align*}
where the first inequality follows from Lemma~\ref{lem:monotonicity} with the fact that $w^-\leq w\leq w^+$ for any $w\in W$ 
and the optimality of $S^*_{w^+}$ in terms of the edge-weight vector $w^+$. 
\end{proof}

This lower bound on the robust ratio is significantly better than the upper bound presented in Theorem~\ref{thm:negative}. 
The upper bound in Theorem~\ref{thm:negative} becomes zero as $n$ increases, whereas $\frac{1}{\max_{e\in E}\frac{r_e}{l_e}}$ does not. 
However, it should be noted that Theorem~\ref{thm:robust_positive} does not contradict Theorem~\ref{thm:negative} 
because Theorem~\ref{thm:robust_positive} supposes that $\min_{e\in E}l_e >0$ holds. 

The following theorem indicates that 
the lower bound on the robust ratio achieved by Algorithm~\ref{alg:basic} is the best possible except for the constant factor. 
\begin{theorem}
There exists an instance $G=(V,E)$ with $W=\times_{e\in E}[l_e,r_e]$ of the robust densest subgraph problem (Problem~\ref{prob:rds}) that satisfies $\min_{e\in E}l_e>0$ 
for which any (deterministic) algorithm returns $S\subseteq V$ that has a robust ratio of $O\left(\frac{1}{\max_{e\in E}\frac{r_e}{l_e}}\right)$. 
\end{theorem}
\begin{proof}
Let $G=(V,E)$ be any graph in which every vertex has degree of at least one and for any $S\subseteq V$, it holds that $|E(S)|\leq \alpha |S|$ for some constant $\alpha$. 
We take $W=\times_{e\in E}[l_e,r_e]$ that satisfies the following three conditions: (i) $l_e=l$ and $r_e=r$ for some $l$ and $r$, respectively; 
(ii) $l>0$; and (iii) $n\geq r/l$. 
Note that any deterministic algorithm for Problem~\ref{prob:rds} returns some $S\subseteq V$. 

For any $S\subsetneq V$, there exists an edge $e'\in E\setminus E(S)$.  
We can construct an edge-weight vector $w'=(w'_e)_{e\in E}$ such that for each $e\in E$, 
\begin{align*}
w'_e=
\begin{cases}
r  &\text{if } e=e',\\
l  &\text{otherwise}. 
\end{cases}
\end{align*}
As $w'\in W$ holds, the robust ratio of $S$ can be upper bounded as follows: 
\begin{align*}
\min_{w\in W}\frac{f_w(S)}{f_w(S^*_w)}
&\leq \frac{f_{w'}(S)}{f_{w'}(S^*_{w'})}\\
&\leq \frac{\alpha |S|\cdot l/|S|}{r/2}
< \frac{2\alpha}{r/l} 
= O\left(\frac{1}{r/l}\right), 
\end{align*}

On the other hand, let $S=V$. 
Select an arbitrary edge $e'\in E$. 
We can again construct an edge weight $w'=(w'_e)_{e\in E}$ such that for each $e\in E$, 
\begin{align*}
w'_e=
\begin{cases}
r  &\text{if } e=e',\\
l  &\text{otherwise}. 
\end{cases}
\end{align*}
As $w'\in W$ again holds, the robust ratio of $S\ (=V)$ can be evaluated as follows: 
\begin{align*}
\min_{w\in W}\frac{f_w(S)}{f_w(S^*_w)}
&\leq \frac{f_{w'}(S)}{f_{w'}(S^*_{w'})}
\leq \frac{((\alpha n-1)l+r)/n}{r/2}\\
&< \frac{\alpha l+r/n}{r/2}
\leq \frac{2(\alpha+1)}{r/l} 
= O\left(\frac{1}{r/l}\right), 
\end{align*}
where the last inequality follows from the fact that $n\geq r/l$ holds. 
Since $r/l=\max_{e\in E}\frac{r_e}{l_e}$ holds, we have the theorem. 
\end{proof}

\section{Robust Densest Subgraph Problem with Sampling Oracle}\label{sec:robust_sampling}
In this section, we formulate the robust densest subgraph problem with sampling oracle 
and present an algorithm with a strong theoretical performance guarantee.

\subsection{Problem Definition}\label{subsec:robust_sampling_problem}
The lower bound on the robust ratio achieved by Algorithm~\ref{alg:basic} (i.e., $\frac{1}{\min_{e\in E}\frac{r_e}{l_e}}$) 
is still small, 
although it is the best possible except for the constant factor. 
This negative result was caused by the fact that 
in the robust densest subgraph problem, we were \emph{excessively conservative} in evaluating the quality of $S\subseteq V$, 
that is, we aimed to find $S\subseteq V$ that has a relatively large density compared to the optimal value on $G$ 
with \emph{any} edge-weight vector $w\in W$. 
In some real-world applications, 
each confidence interval (i.e., $[l_e,r_e]$ for $e\in E$) may be obtained from repeated sampling of an estimated value of the true edge weight; 
therefore, we conjecture that we can obtain a significantly better lower bound on the robust ratio by using such samplings more sophisticatedly. 


To this end, we now formulate the \emph{robust densest subgraph problem with sampling oracle} as follows. 

\begin{problem}[Robust densest subgraph problem with sampling oracle]\label{prob:sampling}
We are given an undirected graph $G=(V,E)$ with an edge-weight space $W=\times_{e\in E} [l_e,r_e]\subseteq \times_{e\in E}[0,\infty)$, 
wherein the unknown true edge-weight vector $w^\text{true}=(w^\text{true}_e)_{e\in E}$ exists. 
In addition, we have access to an edge-weight sampling oracle that accepts an edge $e\in E$ as input and returns a real value as output, in time $\theta$, 
that was drawn independently from a distribution on $[l_e,r_e]$ in which the expected value is equal to the true edge weight $w^\text{true}_e$. 
Given $\gamma \in (0,1)$, we are asked to find 
\begin{itemize}
\item $W_\text{out}\subseteq W$ that satisfies $\mathrm{Pr}[w^\text{true}\in W_\text{out}]\geq 1-\gamma$ and 
\item $S_\text{out}\subseteq V$ that maximizes the robust ratio under edge-weight space $W_\text{out}$, i.e., 
\end{itemize}
\begin{align*}
\min_{w\in W_\text{out}}\frac{f_w(S_\text{out})}{f_w(S^*_w)}, 
\end{align*}
where $S^*_w\subseteq V$ is a densest subgraph on $G$ with edge-weight vector $w$. 
\end{problem}

Let $(W_\text{out},S_\text{out})$ be an output of Problem~\ref{prob:sampling}. 
Since $w^\text{true}\in W_\text{out}$ holds with a probability of at least $1-\gamma$, the following inequality 
\begin{align*}
\frac{f_{w^\text{true}}(S_\text{out})}{f_{w^\text{true}}(S^*_{w^\text{true}})}
\geq \min_{w\in W_\text{out}}\frac{f_w(S_\text{out})}{f_w(S^*_w)}
\end{align*}
also holds with a probability of at least $1-\gamma$. 
Therefore, if $S_\text{out}$ has an objective function value of $\alpha$, 
we observe that $S_\text{out}$ is an $\alpha$-approximate solution for the densest subgraph problem on $G$ with $w^\text{true}$, 
with a probability of at least $1-\gamma$.

\subsection{Algorithm and Analysis}\label{subsec:robust_sampling_algorithm}
Here, we present an algorithm for Problem~\ref{prob:sampling}, with a strong theoretical performance guarantee. 
Our algorithm first obtains $S^*_{w^-}$, i.e., a densest subgraph on $G$ with extreme edge weight $w^-=(l_e)_{e\in E}$, 
to compute the value of $f_{w^-}(S^*_{w^-})$. 
Then, for each $e\in E$, the algorithm iteratively obtains estimated values for the true edge weight of $e$ 
using a sampling oracle for an appropriate number of times, say $t_e$, which will be defined later. 
Note that $t_e$ is determined using the value of $f_{w^-}(S^*_{w^-})$. 
Using the estimated values, 
the algorithm constructs an edge-weight space $W_\text{out}=\times_{e\in E}[l^\text{out}_e,r^\text{out}_e]\subseteq W$, 
which also depends on the value of $f_{w^-}(S^*_{w^-})$, 
and computes a densest subgraph $S_\text{out}$ on $G$ with extreme edge weight $w_\text{out}^-=(l^\text{out}_e)_{e\in E}$. 
The complete procedure is described in Algorithm~\ref{alg:sampling_rev}. 
Note that our algorithm assumes $\max_{e\in E}l_e>0$. 

\begin{algorithm}[t]
\caption{Algorithm with a sampling oracle}\label{alg:sampling_rev}
\SetKwInOut{Input}{Input} 
\SetKwInOut{Output}{Output} 
\Input{\ $G=(V,E)$ with $W=\times_{e\in E}[l_e,r_e]$ (satisfying $\max_{e\in E}l_e>0$), a sampling oracle, $\gamma\in (0,1)$, and $\epsilon >0$}
\Output{\ $(W_\text{out},S_\text{out})$ such that $W_\text{out}\subseteq W$ and $S_\text{out}\subseteq V$}
$S^*_{w^-}\leftarrow~\text{A densest subgraph on } G$ with $w^-=(l_e)_{e\in E}$\;
\For{each $e\in E$}{
  \uIf{$l_e=r_e$}{
    $l_e^\text{out}\leftarrow l_e$, $r_e^\text{out}\leftarrow r_e$\;
  }
  \Else{
    $t_e\leftarrow \left\lceil \frac{m(r_e-l_e)^2 \ln\left(\frac{2m}{\gamma}\right)}{\epsilon^2\cdot f_{w^-}(S^*_{w^-})^2} \right\rceil$\;
    Call the sampling oracle for $e$ for $t_e$ times and observe $w_e^1,\dots,w_e^{t_e}$\;
    $\hat{p}_e\leftarrow \frac{1}{t_e}\sum_{i=1}^{t_e}w_e^i$\;
    $\delta\leftarrow \frac{\epsilon \cdot f_{w^-}(S^*_{w^-})}{\sqrt{2m}}$\;
    $l_e^\text{out}\leftarrow \max\{l_e,\,\hat{p}_e-\delta\}$, $r_e^\text{out}\leftarrow \min\{r_e,\,\hat{p}_e+\delta\}$\;
  }
}
$W_\text{out}\leftarrow \times_{e\in E} [l_e^\text{out},r_e^\text{out}]$\;
$S_\text{out}\leftarrow~\text{A densest subgraph on } G$ with $w^-_\text{out}=(l^\text{out}_e)_{e\in E}$\; 
\Return $(W_\text{out},S_\text{out})$\;
\end{algorithm}

The following theorem provides the theoretical performance guarantee of Algorithm~\ref{alg:sampling_rev}. 

\begin{theorem}\label{thm:sampling_rev}
Let $G=(V,E)$ with $W=\times_{e\in E}[l_e,r_e]$, a sampling oracle, and $\gamma\in (0,1)$ be an instance of Problem~\ref{prob:sampling}. 
Suppose that $\max_{e\in E}l_e>0$ holds. 
Then, for any $\epsilon>0$, 
Algorithm~\ref{alg:sampling_rev} returns
\begin{itemize}
\item $W_\text{out}\subseteq W$ that satisfies $\mathrm{Pr}[w^\text{true}\in W_\text{out}]\geq 1-\gamma$ and 
\item $S_\text{out}\subseteq V$ that satisfies 
\begin{align*}
\min_{w\in W_\text{out}}\frac{f_w(S_\text{out})}{f_w(S^*_w)}\geq 1-\epsilon
\end{align*}
\end{itemize}
in time pseudo-polynomial in the size of $G$ and $W$, $\theta$, and $1/\epsilon$. 
\end{theorem}

In the proof of the above theorem, we use the following form of Hoeffding bound: 
\begin{fact}[Hoeffding bound; Theorem~2 of Hoeffding~\cite{Hoeffding_63}]\label{fact:hoeffding}
Let $X_1,\dots,X_t$ be independent random variables such that $X_i\in [a_i,b_i]$ holds for any $i=1,\dots,t$. 
Then, for any $c> 0$, it holds that 
\begin{align*}
\mathrm{Pr}&\left[\left|\frac{1}{t}\sum_{i=1}^tX_i - \mathrm{E}\left[\frac{1}{t}\sum_{i=1}^tX_i\right]\right|\geq c\right]\\
&\leq 2\exp\left(\frac{-2t^2c^2}{\sum_{i=1}^t (b_i-a_i)^2}\right). 
\end{align*}
\end{fact}


The following lemma is a key ingredient for establishing our theorem, 
which provides an upper bound on the difference between two density values for $S\subseteq V$: 
one with the edge-weight vector $w_1$ and the other with the edge-weight vector $w_2$, 
using the distance between the two vectors $w_1$ and $w_2$. 

\begin{lemma}\label{lem:diff_sparse}
Let $G=(V,E)$ be an undirected graph. 
Let $w_1$ and $w_2$ be edge-weight vectors such that $\|w_1-w_2\|_{\infty}\leq \beta$ holds. 
Then, for any $S\subseteq V$, it holds that 
\begin{align*}
|f_{w_1}(S)-f_{w_2}(S)|\leq \sqrt{\frac{m}{2}}\cdot \beta. 
\end{align*}
\end{lemma}
\begin{proof}
We first consider the case where $|S|< \sqrt{2m}+1$ holds. We have 
\begin{align*}
|f_{w_1}(S)-f_{w_2}(S)|
&=\frac{|w_1(S)-w_2(S)|}{|S|}
\leq \frac{{|S|\choose 2}\cdot \beta}{|S|}\\
&=\frac{(|S|-1)\beta}{2}
\leq \sqrt{\frac{m}{2}}\cdot \beta.
\end{align*}

Next, we consider the case where $|S|\geq \sqrt{2m}+1$ holds. Since $|E(S)|\leq m$ holds, we have
\begin{align*}
|f_{w_1}(S)-f_{w_2}(S)|
&=\frac{|w_1(S)-w_2(S)|}{|S|}\\
&\leq \frac{|E(S)|\cdot \beta}{|S|}
\leq \frac{m\cdot \beta}{|S|}
\leq \sqrt{\frac{m}{2}}\cdot \beta. 
\end{align*}
Thus, we have the lemma. 
\end{proof}

It should be noted that under an assumption identical to that in Lemma~\ref{lem:diff_sparse}, 
we can also obtain an upper bound depending on the size of $S\subseteq V$, 
i.e., $|f_{w_1}(S)-f_{w_2}(S)|\leq \frac{(|S|-1)}{2}\cdot \beta$, 
which is more effective than the upper bound presented in Lemma~\ref{lem:diff_sparse} in the case where $|S|< \sqrt{2m}+1$ holds. 
However, when we do not know the size of $S$ (i.e., we have only $|S|\leq n$), 
the upper bound presented in Lemma~\ref{lem:diff_sparse} is significantly more effective in practice 
because most real-world graphs are sparse, i.e., $m=O(n)$ holds. 
Note that such upper bounds affect the definition of $t_e$ in Algorithm~\ref{alg:sampling_rev}. 
If we use the above upper bound depending on the size of $S$ alternatively, 
we have $t_e= \left\lceil \frac{(n-1)^2(r_e-l_e)^2 \ln\left(\frac{2m}{\gamma}\right)}{2\epsilon^2\cdot f_{w^-}(S^*_{w^-})^2} \right\rceil$, 
which is significantly less effective than ours in practice. 

%

We are now in a position to prove the theorem. 

\begin{proof}[Proof of Theorem~\ref{thm:sampling_rev}]
From the definitions of $l^\text{out}_e$ and $r^\text{out}_e$, we observe that $W_\text{out}\subseteq W$ holds. 
First, we prove that $\mathrm{Pr}[w^\text{true}\in W_\text{out}]\geq 1-\gamma$ holds. 
For any $e\in E$ with $l_e=r_e$, we have $w^\text{true}_e=l_e\ (=r_e)$. 
On the other hand, for any $e\in E$ with $l_e < r_e$, we have 
\begin{align*}
\mathrm{Pr}&\left[w^\text{true}_e \notin [l_e^\text{out},\,r_e^\text{out}]\right]\\
&=\mathrm{Pr}\left[w^\text{true}_e\notin [\max\{l_e,\,\hat{p}_e-\delta\},\,\min\{r_e,\,\hat{p}_e+\delta\}]\right]\\
&=\mathrm{Pr}\left[\left|\frac{1}{t_e}\sum_{i=1}^{t_e}w_e^i - w^\text{true}_e\right|> \frac{\epsilon \cdot f_{w^-}(S^*_{w^-})}{\sqrt{2m}}\right]\\
&\leq 2\exp\left(\frac{-2t_e^2\cdot \frac{\epsilon^2\cdot f_{w^-}(S^*_{w^-})^2}{2m}}{t_e(r_e-l_e)^2}\right)\\
&\leq 2\exp\left(-\ln\left(\frac{2m}{\gamma}\right)\right)
=\frac{\gamma}{m}. 
\end{align*}
The second equality follows from the definition of $\delta$ in the algorithm and the fact that $w^\text{true}_e\in [l_e,r_e]$. 
The first inequality follows from Fact~\ref{fact:hoeffding}, 
and the second inequality follows from the definition of $t_e$ in the algorithm. 
By a union bound, we have 
\begin{align*}
\mathrm{Pr}\left[\exists e\in E,\ w^\text{true}_e\notin [l_e^\text{out},\,r_e^\text{out}]\right]
\leq \frac{\gamma}{m}\cdot m = \gamma,
\end{align*}
which guarantees that $\mathrm{Pr}[w^\text{true}\in W_\text{out}]\geq 1-\gamma$ holds. 

Next, we establish that the output $S_\text{out}\subseteq V$ of Algorithm~\ref{alg:sampling_rev} 
has a robust ratio of at least $1-\epsilon$ under edge-weight space $W_\text{out}$. 
Recall that $w^-_\text{out}=(l_e^\text{out})_{e\in E}$ and $w^+_\text{out}=(r_e^\text{out})_{e\in E}$. 
Noticing that $\|w^+_\text{out}-w^-_\text{out}\|_\infty \leq 2\delta = \sqrt{\frac{2}{m}}\cdot \epsilon \cdot f_{w^-}(S^*_{w^-})$, 
we have 
\begin{align*}
\frac{f_{w^-_\text{out}}(S^*_{w^+_\text{out}})}{f_{w^+_\text{out}}(S^*_{w^+_\text{out}})}
&= 1- 
\frac{f_{w^+_\text{out}}(S^*_{w^+_\text{out}}) - f_{w^-_\text{out}}(S^*_{w^+_\text{out}})}{f_{w^+_\text{out}}(S^*_{w^+_\text{out}})}\\
&\geq 1-\frac{\sqrt{\frac{m}{2}}\cdot \sqrt{\frac{2}{m}}\cdot \epsilon \cdot f_{w^-}(S^*_{w^-})}{f_{w^+_\text{out}}(S^*_{w^+_\text{out}})}\\
&\geq 1-\frac{\epsilon \cdot f_{w^-}(S^*_{w^-})}{f_{w^-}(S^*_{w^-})}
= 1-\epsilon, 
\end{align*}
where the first inequality follows from Lemma~\ref{lem:diff_sparse} 
with the fact that $\|w^+_\text{out}-w^-_\text{out}\|_\infty \leq \sqrt{\frac{2}{m}}\cdot \epsilon \cdot f_{w^-}(S^*_{w^-})$, 
and the second inequality follows from the optimality of $S^*_{w^+_\text{out}}$ in terms of the edge-weight vector $w^+_\text{out}$ 
and Lemma~\ref{lem:monotonicity} with the fact that $w^-\leq w^+_\text{out}$. 

The output $S_\text{out}\subseteq V$ of Algorithm~\ref{alg:sampling_rev} is actually $S^*_{w^-_\text{out}}\subseteq V$. 
Using the above inequality, we can evaluate the robust ratio of $S_\text{out}$ under edge-weight space $W_\text{out}$ as follows: 
\begin{align*}
\min_{w\in W_\text{out}}\frac{f_w(S^*_{w^-_\text{out}})}{f_w(S^*_w)}
\geq \frac{f_{w^-_\text{out}}(S^*_{w^-_\text{out}})}{f_{w^+_\text{out}}(S^*_{w^+_\text{out}})}
\geq \frac{f_{w^-_\text{out}}(S^*_{w^+_\text{out}})}{f_{w^+_\text{out}}(S^*_{w^+_\text{out}})}
\geq 1-\epsilon. 
\end{align*}

Finally, it is evident that Algorithm~\ref{alg:sampling_rev} runs in time pseudo-polynomial in the size of $G$ and $W$, $\theta$, and $1/\epsilon$. 
\end{proof}

%

\section{Experimental Evaluation}\label{sec:experiments}
The purpose of our experiments is to evaluate the effectiveness of our proposed algorithms 
(i.e., Algorithms~\ref{alg:basic} and~\ref{alg:sampling_rev}) in terms of the quality of solutions and computation time. 
To this end, we compare our algorithms with a certain baseline algorithm 
using both synthetic graphs and popular real-world graphs. 
The baseline algorithm, denoted by \textsf{Random}, first selects $w^\text{rand}$ from $W$ uniformly at random; 
then, it returns a densest subgraph on $G$ with edge weight $w^\text{rand}$. 

All the algorithms we compare need to compute a densest subgraph on $G$ with some edge weight $w$. 
To this end, we employed Charikar's LP-based algorithm 
in conjunction with Balalau et al.'s preprocessing, which was described in Section~\ref{sec:preliminaries}. 
To solve the LP relaxations, we used a state-of-the-art mathematical programming solver, Gurobi Optimizer 7.5.1, 
with default parameter settings except for $\texttt{Method}=1$; 
it stipulates that the LP relaxations are solved using a dual simplex algorithm.  

The experiments were conducted on a Linux machine with Intel Xeon Processor E5-2690 v4 2.6~GHz CPU and 256~GB RAM. 
The code was written in Python, which is publicly available.\footnote{https://github.com/atsushi-miyauchi/robust-densest-subgraph-discovery} 

\subsection{Synthetic Graphs}
Here, we report the results of the computational experiments with synthetic graphs. 
To generate synthetic graphs appropriate for our experimental evaluation, we introduce a random graph model, 
which we refer to as the \emph{planted uncertain dense subgraph model}. 

In this model, we first generate an Erd\H{o}s--R\'enyi random graph $G=(V,E)$ with $n$ vertices and edge probability $p$. 
Then, we focus on a subset of vertices $S'\subseteq V$ consisting of $n'\, (\leq n)$ vertices 
as a planted dense region. 
On this graph $G=(V,E)$, we make an edge-weight space $W=\times_{e\in E}[l_e,r_e]$ as follows: 
Let $\alpha\in [0.0, 0.9]$ be a real parameter. 
For each $e\in E$, we set 
\begin{align*}
[l_e,r_e]=
\begin{cases}
[\texttt{rand}(0.1+\alpha,1.0),\, 1.0]  &\text{if }\,e\in E(S'),\\
[0.1,\, \texttt{rand}(0.1,1.0-\alpha)]   &\text{if }\,e\in E\setminus E(S'),
\end{cases}
\end{align*}
where $\texttt{rand}(\cdot,\cdot)$ is a value 
selected uniformly at random from the closed interval between the two values within the parenthesis. 
Note that the larger the parameter $\alpha$, the more significant the difference between $[l_e,r_e]$ for $e\in E(S')$ and $[l_e,r_e]$ for $e\in E\setminus E(S')$.
For example, when $\alpha =0.0$, each $e\in E(S')$ has $[\texttt{rand}(0.1,1.0),\, 1.0]$ and each $e\in E\setminus E(S')$ has $[0.1,\,\texttt{rand}(0.1,1.0)]$; 
however, when $\alpha=0.9$, each $e\in E(S')$ has $[l_e,r_e]=[1.0,1.0]$ and each $e\in E\setminus E(S')$ has $[l_e,r_e]=[0.1,0.1]$. 
Furthermore, we define a true edge-weight vector $w^\text{true}=(w^\text{true}_e)_{e\in E}$ as follows: 
For each $e\in E$, we set 
\begin{align*}
w^\text{true}_e=
\begin{cases}
\texttt{rand}(\max\{l_e,0.9\},\, 1.0)  &\text{if }\,e\in E(S'),\\
\texttt{rand}(0.1,\, \min\{r_e,0.2\})  &\text{if }\,e\in E\setminus E(S'). 
\end{cases}
\end{align*}
More or less, $w^\text{true}_e$ tends to exhibit a relatively large value for $e\in E(S')$ 
and a relatively small value for $e\in E\setminus E(S')$. 
Note that $w^\text{true}\in W$ holds. 

Algorithm~\ref{alg:sampling_rev} requires a sampling oracle, which we simulate as follows: 
For each $e\in E$, the sampling oracle returns $\texttt{rand}(w^\text{true}_e-\texttt{min}_e,\,w^\text{true}_e+\texttt{min}_e)$, 
where $\texttt{min}_e = \min\{w^\text{true}_e-l_e,\, r_e-w^\text{true}_e\}$. 
It should be noted that for every $e\in E$, 
the expected value is equal to the true edge weight $w^\text{true}_e$, as required. 

Throughout our experiments, we set $n=500$ and $p=0.01$. 
In these parameter settings, we construct four types of instances with $n'=50$, $100$, $150$, and $200$; 
in each of these, the parameter $\alpha$ varies from $0.0$ to $0.9$ with increments of $0.1$. 

\begin{figure}[t!]
\centering
\subfigure[$n'=50$]{\includegraphics[width=0.49\columnwidth]{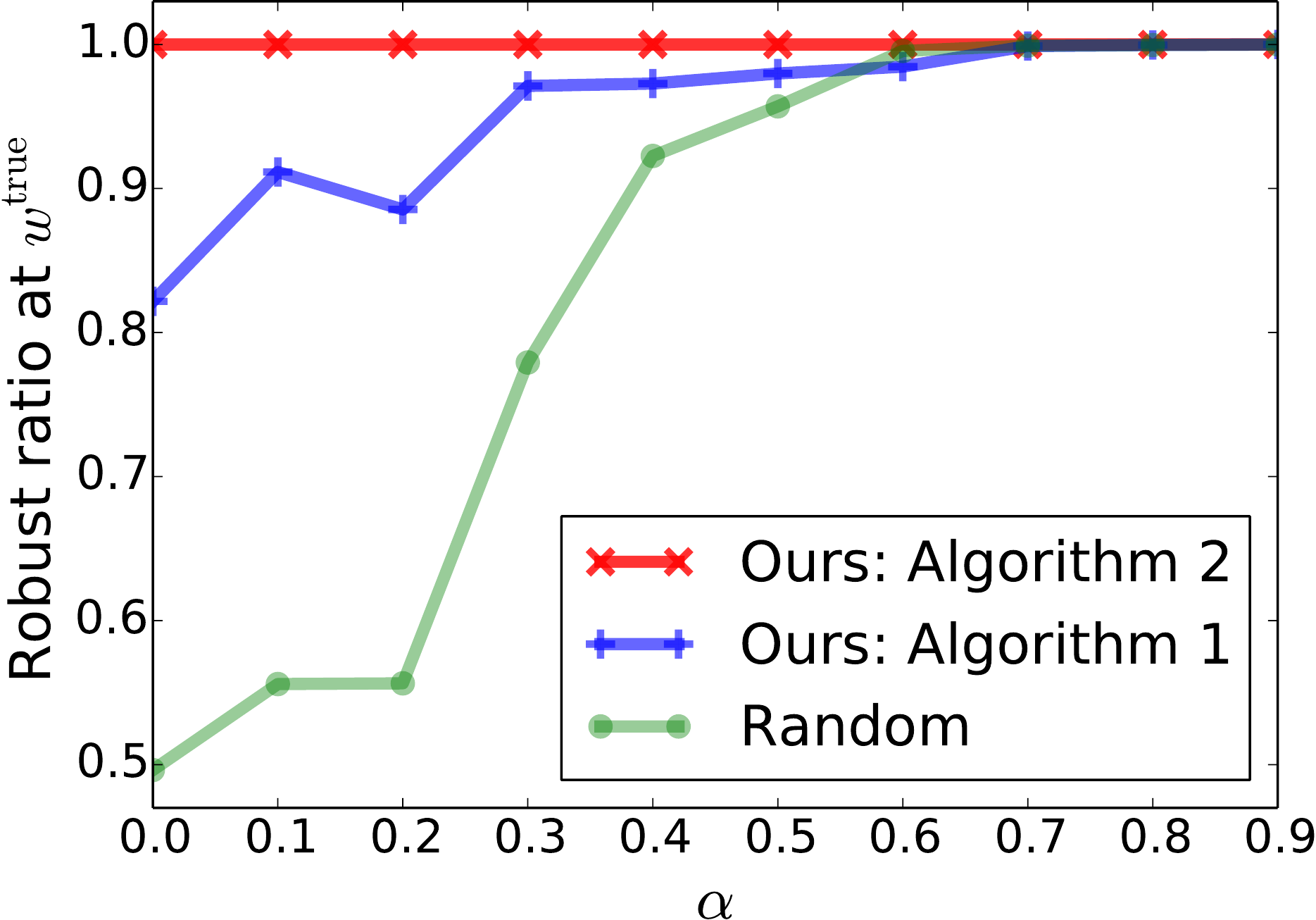}}
\subfigure[$n'=100$]{\includegraphics[width=0.49\columnwidth]{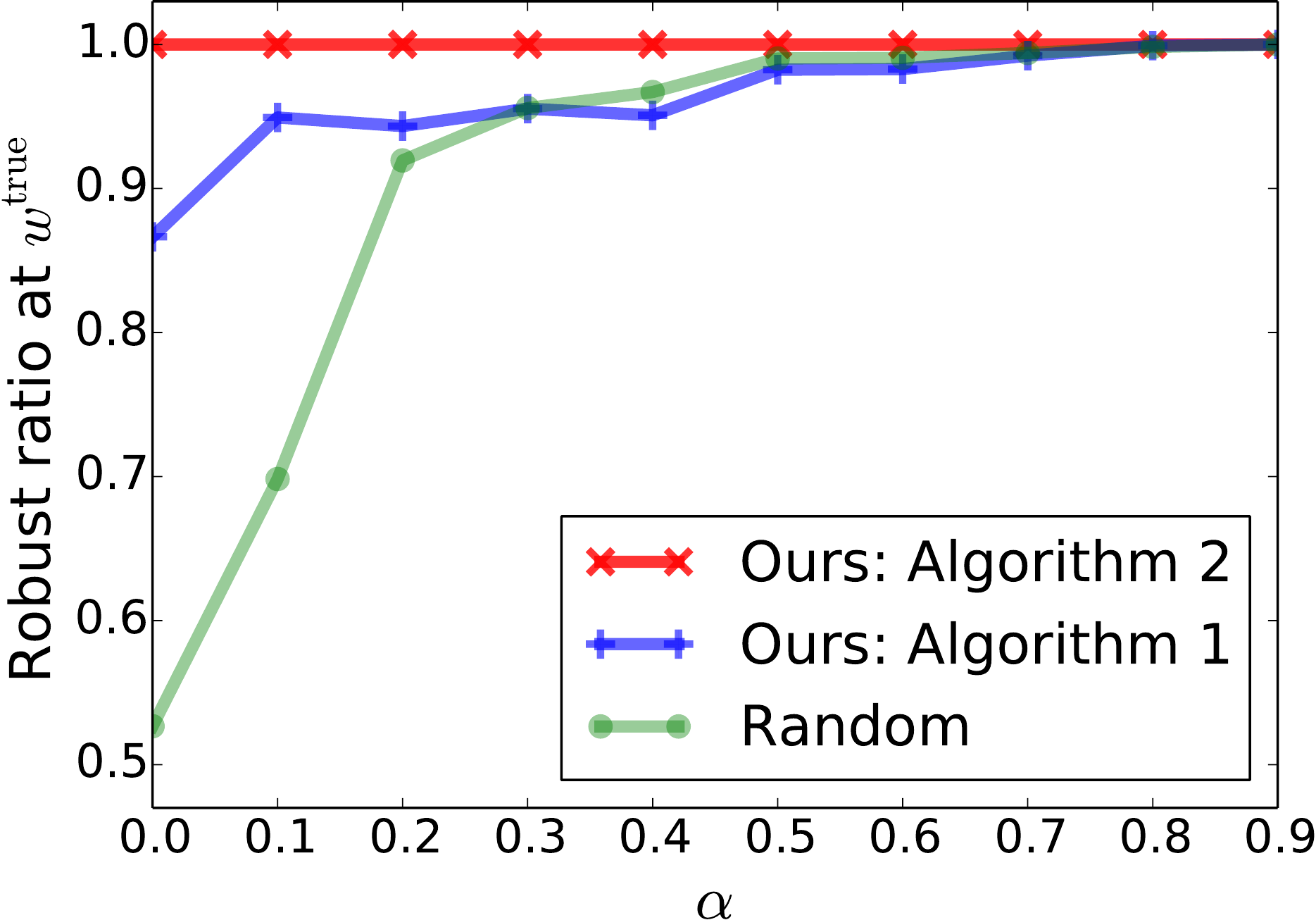}}\\
\subfigure[$n'=150$]{\includegraphics[width=0.49\columnwidth]{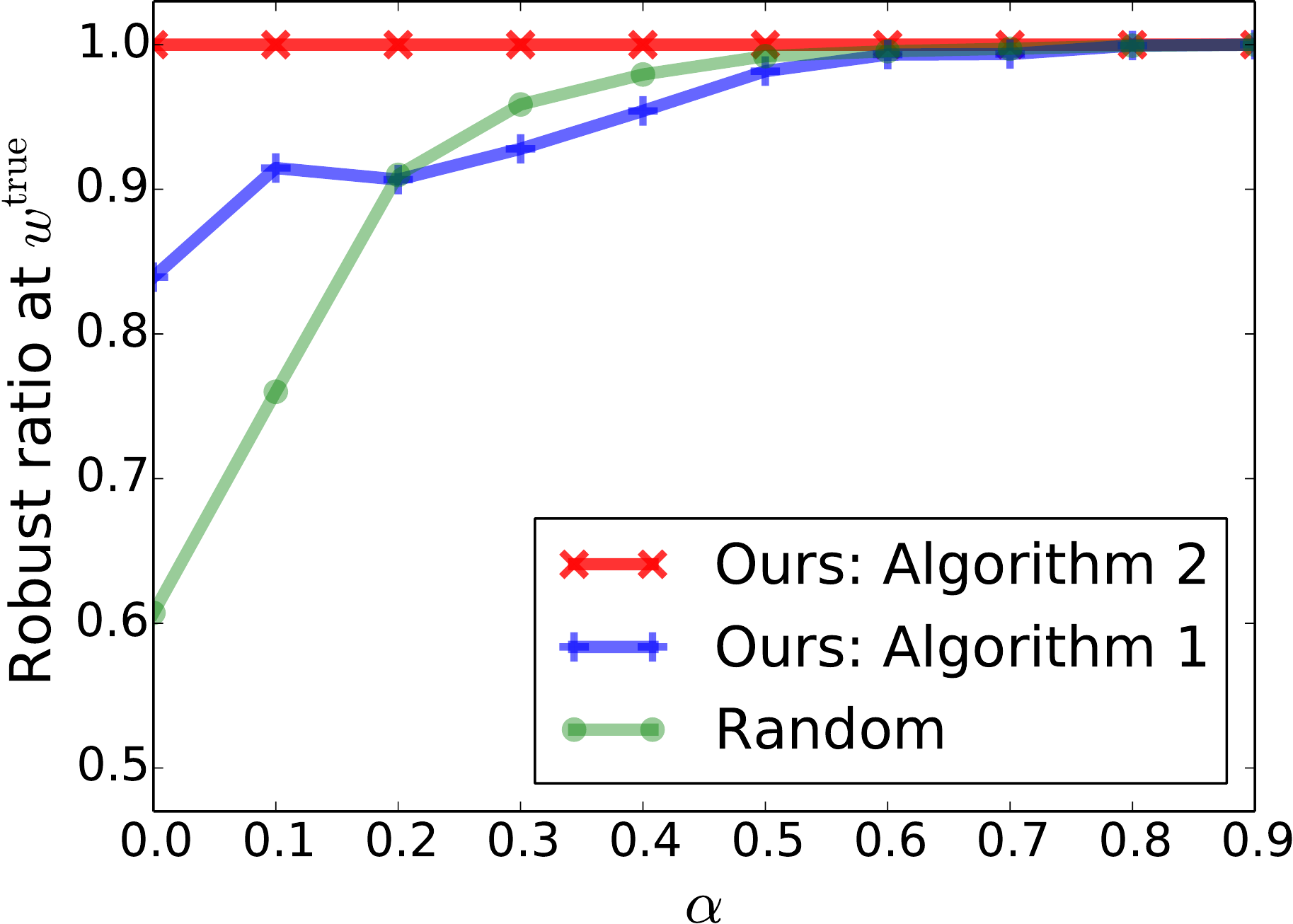}}
\subfigure[$n'=200$]{\includegraphics[width=0.49\columnwidth]{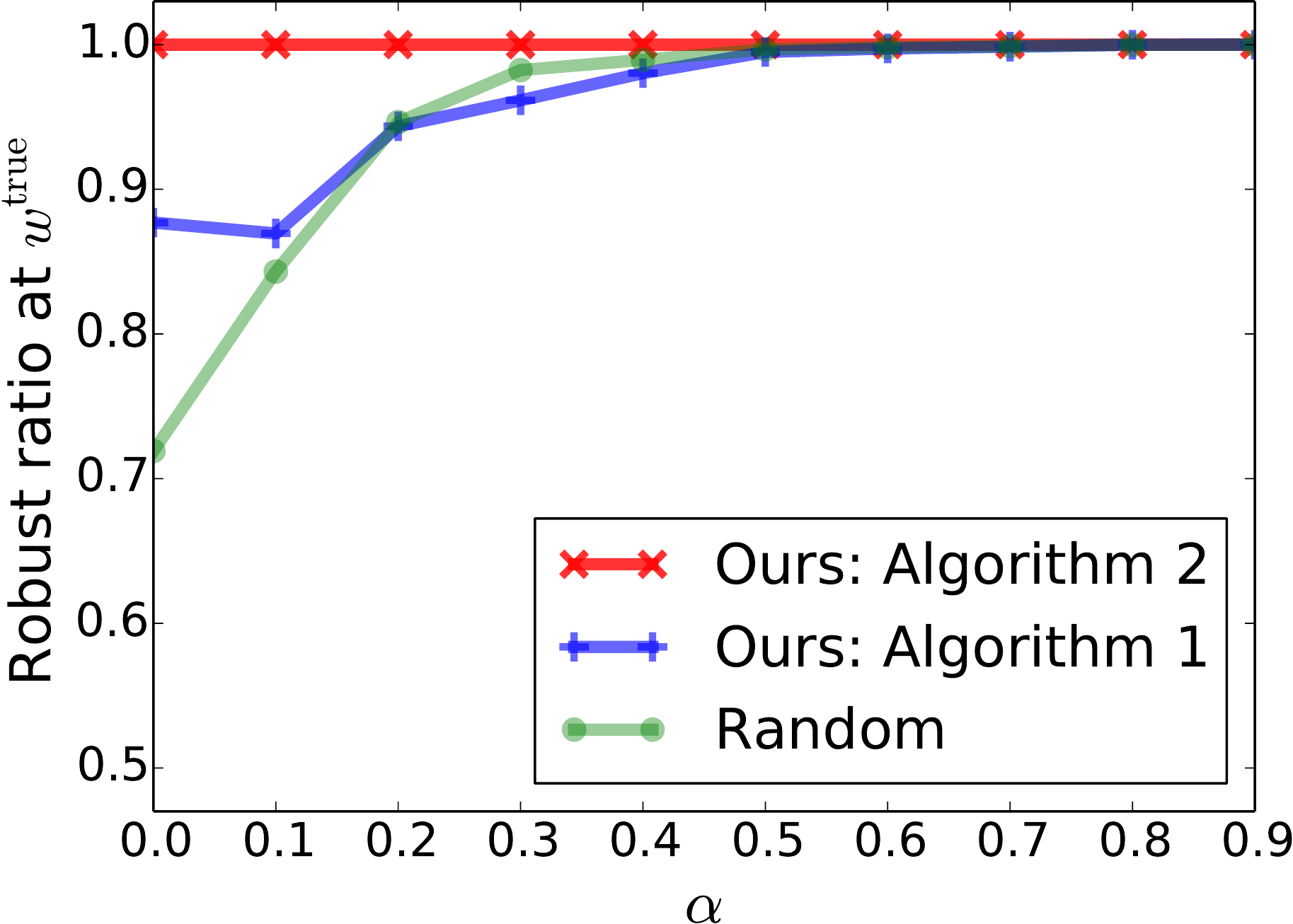}}
\caption{Results for synthetic graphs. Each point corresponds to the average value over 10 graph realizations.}
\label{fig:result_synthetic}
\end{figure}


The results are shown in Figure~\ref{fig:result_synthetic}. 
The quality of output $S\subseteq V$ is evaluated by the robust ratio at $w^\text{true}$, 
i.e., $f_{w^\text{true}}(S)/f_{w^\text{true}}(S^*_{w^\text{true}})$. 
With regard to the parameters in Algorithm~\ref{alg:sampling_rev}, we set $(\gamma,\epsilon)=(0.1, 0.5)$. 
Because \textsf{Random} and Algorithm~\ref{alg:sampling_rev} contain randomness, 
we performed them 10 times for each graph realization 
and considered the average value of the robust ratio at $w^\text{true}$ as the result for the graph. 

As is evident, our proposed algorithms, Algorithms~\ref{alg:basic} and~\ref{alg:sampling_rev}, 
outperform the baseline algorithm \textsf{Random}. 
In particular, owing to the power of the use of a sampling oracle, 
Algorithm~\ref{alg:sampling_rev} obtains $S\subseteq V$ with a significantly high robust ratio; 
the robust ratio almost always attains its upper bound (i.e., 1.0), 
which implies that the output of Algorithm~\ref{alg:sampling_rev} is (almost always) a densest subgraph on $G$ with $w^\text{true}$. 
Note that such a high performance of Algorithm~\ref{alg:sampling_rev} is not a trivial outcome 
because we set $(\gamma,\epsilon)=(0.1, 0.5)$. 
Algorithm~\ref{alg:basic} outperforms \textsf{Random}, particularly under relatively challenging instances with small $\alpha$, 
although it exhibits inferior performance for relatively easy instances with large $\alpha$. 
Both Algorithm~\ref{alg:basic} and \textsf{Random} have higher performances for instances with larger $n'$; 
this appears to be a result of the fact that a planted dense region becomes significant as $n'$ increases.

\subsection{Real-World Graphs}
Here, we report the results of the computational experiments with real-world graphs. 
Table~\ref{tab:real_instance} lists the real-world graphs on which our computational experiments were conducted; 
most of these are available in Leskovec and Krevl~\cite{Leskovec_Krevl_14}. 
As is evident, every graph is sparse, i.e., the average degree $\frac{2m}{n}$ is small. 
Note that all the graphs here were made simple and undirected (if necessary) 
by omitting the directions of the edges and by removing self-loops and redundant multiple edges. 
Furthermore, if a graph is not connected, we take only the largest connected component in the graph. 
To effectively evaluate the robustness of the algorithms, 
we introduce a random model for constructing an edge-weight space and a true edge-weight vector on a (real-world) graph; 
we call this model the \emph{knockout densest subgraph model}.

\begin{table*}[t]
\begin{center}
\caption{Real-world graphs used in our experiments. 
Here, $|S^*|$ denotes the size of a densest subgraph on a graph (with unweighted edges).}\label{tab:real_instance} 
\begin{tabular}{lrr@{\hspace{4.0ex}}rrl}
\toprule
Name          & $n$   & $m$   &$\frac{2m}{n}$  &$|S^*|$ & Description\\
\midrule
\textsf{Karate           } &      34     &       78       &4.59   &16       &Social network  \\ 
\textsf{Lesmis           } &      77     &      254       &6.60   &23       &Co-appearance network  \\ 
\textsf{Polbooks         } &     105     &      441       &8.40   &24       &Co-purchased network  \\ 
\textsf{Adjnoun          } &     112     &      425       &7.59   &48       &Word adjacency   \\ 
\textsf{Football         } &     115     &      613       &10.66   &115     &Sports game network\\ %
\textsf{Jazz             } &     198     &     2,742      &27.70   &100     &Social network     \\ 
\textsf{email-Eu-core    } &     986     &     16,064     &32.58   &224     &Email communication    \\ 
\textsf{Email            } &    1,133    &     5,451      &9.62   &301      &Email communication    \\ 
\textsf{Polblogs         } &    1,222    &    16,714      &27.36   &139     &Blog network     \\ 
\textsf{Wiki-Vote        } &    7,066    &   100,736      &28.51   &835     &Wikipedia ``who-votes-whom''     \\ 
\textsf{ca-HepTh         } &    8,638    &    24,806      &5.74   &32       &Co-authorship network  \\ 
\textsf{ca-HepPh         } &   11,204    &   117,619      &21.00   &239     &Co-authorship network    \\ 
\textsf{ca-CondMat       } &   21,363    &    91,286      &8.55   &30       &Co-authorship network  \\ 
\textsf{AS-22july06      } &   22,963    &    48,436      &4.22   &104      &Internet at autonomous system level   \\%
\textsf{email-Enron      } &   33,696    &   180,811      &10.73   &555     &Email communication     \\ 
\textsf{web-Stanford     } &  255,265    &  1,941,926   &15.21   &597       &Web graph   \\%
\textsf{web-NotreDame    } &  325,729    &  1,090,108   &6.69   &1,367      &Web graph \\%
\bottomrule
\end{tabular}
\end{center}
\end{table*}

\begin{table*}[t]
\begin{center}
\caption{Results for real-world graphs.}\label{tab:real_result}
\begin{tabular}{lrrrrrrrrrrrrrr}
\toprule
Name
&\multicolumn{2}{c}{\textsf{Random}} 
&&\multicolumn{2}{c}{Algorithm~\ref{alg:basic}}   
&&\multicolumn{3}{c}{Algorithm~\ref{alg:sampling_rev}}\\
\cline{2-3}
\cline{5-6}
\cline{8-10}
&Ratio&Time(s)
&&Ratio&Time(s)
&&Ratio&Time(s)&\#Calls (avg.)\\
\midrule
\textsf{Karate          } &0.992    &0.00    &&1.000  &0.00    &&1.000       &0.01   &92.65  \\
\textsf{Lesmis          } &0.992    &0.01    &&1.000  &0.01    &&1.000       &0.03   &77.12 \\
\textsf{Polbooks        } &0.992    &0.02    &&0.980  &0.01    &&1.000       &0.06   &99.97 \\
\textsf{Adjnoun         } &0.908    &0.02    &&0.958  &0.02    &&1.000       &0.15   &340.66  \\
\textsf{Football        } &0.995    &0.04    &&1.000  &0.04    &&1.000       &2.05   &3991.49  \\
\textsf{Jazz            } &0.999    &0.15    &&0.990  &0.08    &&1.000       &0.38   &118.86  \\
\textsf{email-Eu-core   } &0.960    &3.02    &&0.994  &2.35    &&1.000       &11.48   &492.59  \\
\textsf{Email           } &0.851    &0.82    &&0.980  &0.69    &&1.000       &8.88    &1918.53  \\
\textsf{Polblogs        } &0.999    &1.89    &&0.997  &1.46    &&1.000       &5.33    &186.84  \\
\textsf{Wiki-Vote       } &0.961    &50.45   &&0.994  &54.27    &&1.000       &201.25  &1359.74\\
\textsf{ca-HepTh        } &1.000    &1.29    &&1.000  &1.34    &&1.000       &7.86     &677.20 \\
\textsf{ca-HepPh        } &0.713    &32.85    &&0.995  &19.70    &&1.000       &57.05  &339.81 \\
\textsf{ca-CondMat      } &0.999    &9.03    &&0.998  &8.76    &&1.000       &142.98  &2751.65  \\
\textsf{AS-22july06     } &0.940    &2.02    &&0.987  &2.25    &&1.000       &8.40    &623.54  \\
\textsf{email-Enron     } &0.952    &106.20    &&0.998  &179.63    &&1.000       &374.16  &1613.85  \\
\textsf{web-Stanford    } &0.998    &55.27    &&0.993  &47.19    &&1.000       &419.49    &1129.58  \\
\textsf{web-NotreDame   } &1.000    &482.26    &&0.999  &125.07    &&1.000       &760.15  &2378.03  \\
\bottomrule
\end{tabular}
\end{center}
\end{table*}

First, we explain the intuition behind the model. 
Let $G=(V,E)$ be a given (real-world) undirected graph and $S^*\subseteq V$ be a densest subgraph on $G$ (with unweighted edges). 
Suppose here that we put a very small true edge weight $t_e$ for each $e\in E(S^*)$, 
whereas we put a relatively large true edge weight $t_e$ for each $e\in E\setminus E(S^*)$. 
Suppose also that the edge-weight space only marginally reflects the values of the true edge weights. 
In such a situation, from the structure (i.e., the existence/non-existence of edges) of the graph, 
any algorithm that does not consider the edge-weight space or sampling oracle with adequate caution tends to detect $S^*$ 
despite the fact that $S^*$ is no longer likely to be a densest subgraph on $G$ with $w^\text{true}$. 

The knockout densest subgraph model is a random model that simulates the above situation. 
Specifically, we make an edge-weight space $W=\times_{e\in E}[l_e,r_e]$ as follows: 
For each $e\in E$, we set 
\begin{align*}
[l_e,r_e]=
\begin{cases}
[0.1,\, \texttt{rand}(0.1,0.9)]   &\text{if }\,e\in E(S^*),\\
[\texttt{rand}(0.2,1.0),\, 1.0]  &\text{if }\,e\in E\setminus E(S^*). 
\end{cases}
\end{align*}
In addition, we define a true edge-weight vector $w^\text{true}=(w^\text{true}_e)_{e\in E}$ as follows: 
For each $e\in E$, we set 
\begin{align*}
w^\text{true}_e=
\begin{cases}
\texttt{rand}(0.1,\, \min\{r_e,0.11\})  &\text{if }\,e\in E(S^*),\\
\texttt{rand}(\max\{l_e,0.99\},\, 1.0)  &\text{if }\,e\in E\setminus E(S^*).
\end{cases}
\end{align*}
Note that $w^\text{true}\in W$ holds. 
Algorithm~\ref{alg:sampling_rev} requires a sampling oracle, 
which we simulate in a manner identical to that in the planted uncertain dense subgraph model. 

The results are summarized in Table~\ref{tab:real_result}. 
The quality of output $S\subseteq V$ is again evaluated by the robust ratio at $w^\text{true}$. 
To observe the scalability, we list the computation time for the algorithms. 
With regard to Algorithm~\ref{alg:sampling_rev}, we also list the average number of calls of the sampling oracle per edge. 
With regard to the parameters in Algorithm~\ref{alg:sampling_rev}, 
to apply the algorithm to large graphs, we set $(\gamma,\epsilon)=(0.9,0.9)$. 
Moreover, we perform a simple preprocessing algorithm, which was inspired by Balalau et al.'s preprocessing technique, 
to reduce the size of a given graph. 
This preprocessing does not impair the theoretical performance guarantee of our algorithm. 
Owing to space limitations, we omit the details here. 
With regard to \textsf{Random} and Algorithm~\ref{alg:sampling_rev}, 
we performed them 10 times on each graph 
and considered the average value of each of the robust ratio at $w^\text{true}$ and the computation time 
as the results for the graph. 

As is evident, the trend is consistent with the results of the experiments with synthetic graphs; 
that is, Algorithms~\ref{alg:basic} and~\ref{alg:sampling_rev} outperform \textsf{Random}. 
Algorithm~\ref{alg:sampling_rev} (almost always) obtains a densest subgraph on $G$ with $w^\text{true}$; 
Algorithm~\ref{alg:basic} outperforms \textsf{Random}, particularly in relatively challenging instances 
for which \textsf{Random} only obtains $S\subseteq V$ with a robust ratio of at most 0.95. 
Algorithm~\ref{alg:sampling_rev} is not significantly worse in terms of the scalability.

\section{Conclusion}\label{sec:conclusion}
In this study, we have provided a framework for dense subgraph discovery under the uncertainty of edge weights. 
Specifically, we have addressed such an uncertainty issue using the theory of robust optimization. 
First, we formulated the robust densest subgraph problem (Problem~\ref{prob:rds}) 
and presented a simple algorithm (Algorithm~\ref{alg:basic}). 
We then formulated the robust densest subgraph problem with sampling oracle (Problem~\ref{prob:sampling})
that models dense subgraph discovery using an edge-weight sampling oracle, 
and presented an algorithm with a strong theoretical performance guarantee (Algorithm~\ref{alg:sampling_rev}). 
Computational experiments using both synthetic graphs and popular real-world graphs demonstrated 
the effectiveness of our proposed algorithms. 


\section*{Acknowledgment}
The authors wish to thank the anonymous reviewers for their valuable comments. 
The authors also wish to thank Yuko Kuroki for her helpful comments, which improved the presentation of the paper. 
This work was supported by JST CREST Grant Numbers JPMJCR14D2 and JPMJCR15K5, Japan. 
A.M. is supported by a Grant-in-Aid for Research Activity Start-up (No. 17H07357).

\bibliographystyle{abbrv}
\bibliography{robust_densest_subgraph_discovery}

\end{document}